\documentclass[conference,10pt]{IEEEtran}
\usepackage[dvips]{color}
\usepackage{epsf}
\usepackage{times}
\usepackage{epsfig}
\usepackage{graphicx}

\usepackage{amsmath}
\usepackage{amssymb}
\usepackage{amsxtra}
\usepackage{amsthm}
\usepackage{bbm}

\usepackage{here}
\usepackage{rawfonts}
\usepackage{times}
\usepackage{url}
\usepackage{cite}
\usepackage{comment}
\usepackage[utf8]{inputenc}
\usepackage{caption}
\usepackage{subcaption}

\usepackage{pstricks}
\usepackage{algorithm}
\usepackage{algpseudocode}
\usepackage{lipsum}
\usepackage{mathtools}

\newtheorem{theorem}{\bf Theorem}

\newtheorem{proposition}{\bf Proposition}

\usepackage{geometry}
\makeatletter 
 
\begin{document} 
\title{\huge Distributed Generative Adversarial Networks for mmWave Channel Modeling in Wireless UAV Networks }
\vspace{-0.1cm}
\author{ 
	\IEEEauthorblockN{Qianqian Zhang, 
		Aidin Ferdowsi, 
		and   
		Walid Saad
	}
	
	\IEEEauthorblockA{\small
		Bradley Department of Electrical and Computer Engineering, Virginia Tech, VA, USA,
		Emails: \url{{qqz93,aidin,walids}@vt.edu}. 
	}
}\vspace{-0.2cm}
\maketitle

\setlength{\columnsep}{0.55cm}

\begin{abstract}
In this paper, a novel framework is proposed to enable  air-to-ground  channel modeling  over millimeter wave (mmWave) frequencies in an unmanned aerial vehicle (UAV) wireless network. 
First, an effective channel estimation approach is developed to collect  mmWave channel information allowing each UAV to train a local channel model via a generative adversarial network (GAN).  
Next, in order to share the channel information between UAVs in a privacy-preserving manner,  a cooperative framework, based on a distributed GAN architecture, is developed to enable each UAV to learn the mmWave channel distribution from the entire dataset in a fully distributed  approach.   
The necessary and sufficient conditions for the optimal network structure that maximizes the learning rate for information sharing in the distributed network are derived.   
Simulation results show that the learning rate of the proposed GAN approach will increase by sharing more generated channel samples at each learning iteration, but  decrease given more UAVs in the  network. 
The results also show that the proposed GAN method yields a higher learning accuracy, compared with a standalone GAN, and 
improves the average rate for UAV downlink communications by over $10\%$, compared with a baseline real-time channel estimation scheme. 
 
\end{abstract}

\IEEEpeerreviewmaketitle

\section{Introduction}

Millimeter wave (mmWave) frequencies are a  pillar of next-generation communication systems, and they will support a variety of new applications, such as ultra-high-speed low-latency communications and airborne wireless networks.   
In order to overcome the fast attenuation of mmWave signals, multiple-input multiple-output (MIMO) technologies are often used so as to increase the cell throughput and reduce the multi-user interference. 
Compared with the sub-6 GHz spectrum, the higher frequency of mmWave yields a shorter coherence time for the wireless channels. 
Therefore, mmWave communication links are more time-sensitive and require frequent channel measurements. 
However,  real-time estimation of mmWave MIMO channels can  cause a heavy communication overhead \cite{dong2019deep}. 
In order to  improve the transmission efficiency, it is essential to characterize the mmWave wireless link and accurately model its underlying MIMO channels, thus enabling realistic assessment and effective deployment of mmWave wireless networks.   

Compared with a terrestrial communication network,  mmWave channel modeling  for an airborne, drone-based wireless cellular system is more challenging, due to the mobile location of an aerial base station, and the limited studies on the air-to-ground (A2G) channel characteristics \cite{dabiri2020analytical}.  
Traditional channel modeling methods, such as ray-tracing, becomes very difficult and time-consuming to measure A2G channels, and the generated model cannot be flexibly generalized into other communication environments \cite{yang2019generative}.   
In order to address this challenge, an unmanned aerial vehicle (UAV) base station can collect the A2G channel information during its wireless service, and, then, build a  stochastic  model to estimate the long-term channel parameters. 
The A2G channel model enables the UAV base station to estimate the mmWave link state, thus saving pilot training time and transmit power for efficient communications.  
Therefore,  wireless channel modeling is essential to support a scalable deployment of  mmWave UAV wireless  networks.  

In order to capture the stochastic characteristics of mmWave channels from measurement results,  a number of data-driven modeling approaches were developed in \cite{dong2019deep} and  \cite{akdeniz2014millimeter,han2016two,alkhateeb2019deepmimo,elbir2020federated,park2020communication}.  
Traditional methods, such as spatial-temporal correlation \cite{akdeniz2014millimeter} and compressed sensing \cite{han2016two}, were investigated for characterizing mmWave MIMO transmissions. 
Recent works in \cite{dong2019deep, elbir2020federated} and  \cite{alkhateeb2019deepmimo} developed machine learning methods to extract the propagation feature of mmWave-based communication links. 
The authors in \cite{dong2019deep}  applied deep learning tools for estimating MIMO channels over mmWave frequencies. 
A deep learning dataset  is introduced in \cite{alkhateeb2019deepmimo} for the performance evaluation of mmWave MIMO transmissions. 
However, all of the proposed modeling frameworks  in \cite{dong2019deep} and  \cite{akdeniz2014millimeter,han2016two,alkhateeb2019deepmimo}  depend on the local dataset of a single channel learner, and, thus, the generated channel model is constrained by a limited amount of channel samples and  a few dedicated measurement environments. 
In order to extend the channel model to large-scale application scenarios, a cooperative approach with distributed channel datasets was proposed in \cite{elbir2020federated} to train the channel model   using a federated learning (FL) framework.  
However, the centralized network structure of the FL framework requires a global controller for information aggregation,  and, thus, it cannot operate in a fully distributed network. 
Meanwhile, the FL-based discriminative approach to channel modeling in \cite{elbir2020federated} requires pilot signals to model the accurate channel state information (CSI).   
Furthermore, the work in \cite{park2020communication} characterizes time-varying channel models, by continuously exchanging data in a distributed wireless system.  
However, sharing the raw cellular data in a real-time manner yields a heavy communication overhead, and violates the privacy of mobile users by revealing their location-time information  to other unauthorized entities.  

The main contribution of this paper is a novel framework that can perform  data collection and  channel modeling for mmWave communications in a distributed UAV  network.  
First, an effective channel measurement approach is developed to collect the real-time channel information allowing each UAV to train  a local model  via a generative adversarial network (GAN).   
Next,   to expand the application scenarios of the trained mmWave channel model into a broader spatial-temporal domain, a cooperative learning framework, based on the distributed framework of brainstorming GANs \cite{ferdowsi2020brainstorming}, is developed for each UAV to learn the channel distribution from other agents in a fully distributed manner. 
This generative approach allows us to characterize an underlying  distribution of the mmWave channels based on the entire spatial-temporal domain of measured channel dataset. 
Meanwhile,  to avoid revealing the real measured data or the trained channel model to other agents, each UAV shares synthetic channel samples that are generated from its local  channel model in each  iteration.   
We derive the necessary and sufficient conditions for the optimal network structure  that maximizes the learning rate for information sharing in the distributed network.   
Simulation results show that the learning rate of the proposed GAN approach will increase by sharing more generated channel samples in each  iteration, but it decreases for larger networks.   
The results also show that the proposed GAN approach yields a higher learning accuracy, compared with a standalone GAN without information sharing, and  improves the average data rate of UAV downlink communications by over $10\%$, compared with a baseline real-time channel estimation scheme.

The rest of this paper is organized as follows. 
Section \ref{sysModel} presents the communication model and data collection. 
The UAV network, learning framework, and problem formulation are presented in Section \ref{channelModeling}. 
The optimal network structure and learning solutions are derived in Section \ref{solution}. 
Simulation results are shown in Section \ref{simulation}. 
Conclusions are drawn in Section \ref{conclusion}.

\section{Communication Model and Data Collection }\label{sysModel}


\subsection{Millimeter Wave Channel Model}
Consider an aerial cellular network, in which a set of UAVs provide mmWave downlink communications   to  ground user equipment (UE). 
Each UAV and each UE will be equipped with $M$ and $N$ antennas, respectively. 
The MIMO channel matrix $\boldsymbol{H} \in \mathcal{C}^{N\times M}$  can be given by
$	\boldsymbol{H} = \sum_{k=1}^{K} \alpha_k \boldsymbol{a}_r(\phi_k^r) \boldsymbol{a}_t^H (\phi_k^t)$,  
where  $(\cdot)^H$ is conjugate transpose,  $K$ is the number of  paths, $\alpha_k$ is the  complex  gain of path $k$, and  $\boldsymbol{a}_t(\phi_k^t)$ $ \in \mathcal{C}^{M\times 1} $ and $\boldsymbol{a}_r(\phi_k^r)$ $ \in $  $\mathcal{C}^{N\times 1}$ are the transmit steering vector of angle of departure $\phi_k^t$ and receive  vector of angle of arrival $\phi_k^r$, respectively.   
We assume  uniform linear antenna arrays \cite{han2016two}, 
with the steering vectors given by   
$\boldsymbol{a}_t(\phi^t)$ $=$ $[1, $ $e^{j\frac{\pi}{\lambda} \sin(\phi^t)}, \cdots, e^{j(M-1)\frac{ \pi}{\lambda} \sin(\phi^t)} ]^T $  
and $\boldsymbol{a}_r(\phi^r)$ $ = $ $[1, $ $e^{j\frac{ \pi}{\lambda} \sin(\phi^r)}, $ $ \cdots, $ $e^{j(N-1)\frac{ \pi}{\lambda}\sin(\phi^r)} ]^T $, where $\lambda$ is the carrier wavelength.

Given the fact that the A2G channel via mmWave frequencies has very few scattering paths, we assume that $K=1$ for a line-of-sight (LOS) scenario, i.e. each LOS A2G channel consists of a single path that directly connects the UAV and the UE, while in an non-line-of-sight (NLOS) state scenario, the number of paths is zero.   
Then, for a UAV located at coordinates  $\boldsymbol{x}$ and a UE located at coordinates $\boldsymbol{y}$, the A2G channel model at the service time $t$ can be rewritten as
$\boldsymbol{H}(\boldsymbol{x},\boldsymbol{y},t) =  \alpha (\boldsymbol{x},\boldsymbol{y},t) \boldsymbol{a}_r(\boldsymbol{x},\boldsymbol{y}) \boldsymbol{a}_t^H (\boldsymbol{x},\boldsymbol{y})$,    
where $|\alpha_{\textrm{NLOS}}|=0$ and $|\alpha_{\textrm{LOS}}|>0$. 
Here, the values of $\boldsymbol{a}_r(\boldsymbol{x},\boldsymbol{y})$ and $\boldsymbol{a}_t (\boldsymbol{x},\boldsymbol{y})$ will be uniquely determined by the locations of the UAV-UE pair.  
Then, the estimation of the channel matrix can be obtained by determining the parameter  $\alpha$, in terms of the transmitter's and receiver's locations, as well as the service time. 

\subsection{Channel Estimation and Data Collection}

In order to estimate the A2G mmWave channel, each UAV transmits a pilot symbol with signal power $P$. Let $\boldsymbol{w}$ and $\boldsymbol{q}$ be the beamforming and combining vectors for  channel estimation, respectively. Then, the received pilot signal at the UE is 
\begin{equation}\label{receivedPilot}
	r = \sqrt{P} \boldsymbol{q}^H \boldsymbol{H} \boldsymbol{w} + \boldsymbol{q}^H \boldsymbol{n},
\end{equation}
where  $\boldsymbol{n} \sim \mathcal{CN}(\boldsymbol{0},\sigma^2_{\textrm{UE}}\boldsymbol{I}_N)$ is the noise vector. 
Let $\otimes$ be the Kronecker product, and $\textrm{vec}(\cdot)$ be the  vectorization of a matrix. 
Then, the received pilot signal in (\ref{receivedPilot}) can be rewritten as
\begin{equation}
\begin{aligned}
r &= \sqrt{P} (\boldsymbol{w}^T \otimes \boldsymbol{q}^H) \textrm{vec}(\boldsymbol{H})  + \boldsymbol{q}^H \boldsymbol{n}, \\
&= \sqrt{P} (\boldsymbol{w}^T \otimes \boldsymbol{q}^H) (\boldsymbol{a}_t^{*}  \otimes \boldsymbol{a}_r) \alpha(\boldsymbol{x},\boldsymbol{y},t)  + \boldsymbol{q}^H \boldsymbol{n} , \\
&= \beta \alpha(\boldsymbol{x},\boldsymbol{y},t)  + \boldsymbol{q}^H \boldsymbol{n},  
\end{aligned}
\end{equation} 
where $(\cdot)^{T}$ is transpose, $(\cdot)^{*}$ is complex conjugate,  and $\beta = \sqrt{P} (\boldsymbol{w}^T \otimes \boldsymbol{q}^H) (\boldsymbol{a}_t^{*}  \otimes \boldsymbol{a}_r) \in \mathcal{C}$. 
After receiving $r$, each UE will send the pilot training information to the UAV via  a sub-$6$ GHz uplink \cite{semiari2019integrated}.   
Note that, the beamforming and combining vectors are known by the BS for training purpose.   
Therefore, based on the  pilot signal and  location information, the UAV located at $\boldsymbol{x}$ can estimate the downlink channel gain towards a UE located $\boldsymbol{y}$ at time $t$ via
\begin{equation}
\begin{aligned}
\tilde{\alpha}(\boldsymbol{x},\boldsymbol{y},t) = r \beta^{-1}  =  {\alpha}(\boldsymbol{x},\boldsymbol{y},t) + \tilde{n},
\end{aligned}
\end{equation}  
where $\tilde{n} =  \boldsymbol{q}^H \boldsymbol{n}\beta^{-1} $ is the uncorrelated estimation error.     
During the aerial cellular service, the channel gain  $\tilde{\alpha}$ can be measured and collected by each UAV over a spatial-temporal domain. 
We denote the channel dataset of a given UAV $i$ as $\mathcal{S}_i = \{\boldsymbol{s}_n\}_{n = 1,\cdots,S_i}=\{\boldsymbol{x}_n,\boldsymbol{y}_n,t_n, \tilde{\alpha}_n \}_{n = 1,\cdots,S_i}$, where $ S_i = |\mathcal{S}_i| $ is the number of data samples. 
Based on $\mathcal{S}_i$, each UAV $i$ can build its own model for estimating  A2G mmWave channels in its dedicated measurement area. 
However, over a large spatial-temporal domain, it is very challenging to develop a stochastic model that properly  captures the amplitude and phase coefficients of the MIMO channel response, due to distinct communication environments and a large span of channel parameter values.   
To address this challenge, we  will introduce a deep learning approach to enable   accurate A2G channel modeling over mmWave spectrum.

\section{Channel Modeling via Distributed GANs}\label{channelModeling}

Given the channel dataset $\mathcal{S}_i$, each UAV $i$ can train its own channel model, based on a deep neural network (DNN), to characterize the underlying distribution  $(\boldsymbol{x},\boldsymbol{y},t,\alpha)\sim f_i$. 
This channel distribution enables each UAV to estimate its mmWave link gain 
$\alpha$, while identifying the spatial-temporal range of $(\boldsymbol{x},\boldsymbol{y},t)$ that defines the applicable domain of the trained channel model. 
However, in practice,  each UAV only has a limited amount of channel data samples.  
Thus, a mmWave channel model that is trained based on a local dataset, can be biased and only feasible for a limited spatial-temporal domain. 
Once the UAV moves to an unvisited area,   pilot measurement will again be necessary in order to acquire the propagation feature of the new environment  and update the channel model.  
However, both data collection and model update processes are time-demanding and energy-consuming for a UAV platform.  
Therefore,  to avoid repeated channel estimations within the same area, a UAV can learn the channel information from other UAVs that operated in this region. 
However, raw data exchange in a distributed manner can yield a heavy communication overhead, and it may raise privacy concerns by sharing the location-time information  of mobile UEs to an unauthorized UAV, especially when each UAV  belongs to a different network operator.

\subsection{Distributed GAN Framework: Preliminaries} 
In order to share channel data in a communication-efficient and privacy-preserving approach within the UAV network, a distributed GAN framework is proposed to cooperatively model the A2G mmWave  channel.  
The GAN framework trains a model to generate channel samples from an underlying  distribution of its dataset, without explicitly revealing the data distribution or showing real data samples.   
Given a set $\mathcal{I}$ of $I$ UAVs, we consider that the available data in $ \mathcal{S} =\mathcal{S}_1 \cup \cdots \cup \mathcal{S}_I$  follows a distribution $f$. 
The local dataset $\mathcal{S}_i$ of  each UAV $i$  is collected from different geographic areas or at different service times. 
Hence, each local dataset $\mathcal{S}_i$ follows a distribution $f_i$ that does not span the entire  space of the real channel distribution. 

In a GAN framework, each UAV $i$ has a generator $G_i$, a discriminator $D_i$ and a local dataset $\mathcal{S}_i$.  
The generator $G_i(\boldsymbol{z},\boldsymbol{\theta}^g_{i})$ is a DNN with a parameter vector $\boldsymbol{\theta}^g_i$, which maps  a random input $\boldsymbol{z}$ to the channel sample space  $ \mathcal{S}$, and the discriminator $D_i(\boldsymbol{s}, \boldsymbol{\theta}^d_i)$ is another DNN with a parameter vector $\boldsymbol{\theta}^d_i$ that takes a channel sample  $\boldsymbol{s}$ as an input and outputs a value between $0$ and $1$. 
If the output of $D_i$ is close to $1$, then the input sample $\boldsymbol{s}$ is similar to the real data sample in $\mathcal{S}_i$; otherwise, a zero output of   $D_i$ means that the input data is fake. 
Therefore, the generator of each UAV $i$ aims to generate channel samples close to the real measurement data, while the discriminator tries to distinguish the fake  samples from the real channel samples.

\begin{figure}[!t]
	\begin{center}
		\vspace{-0.12cm}
		\includegraphics[width=9cm]{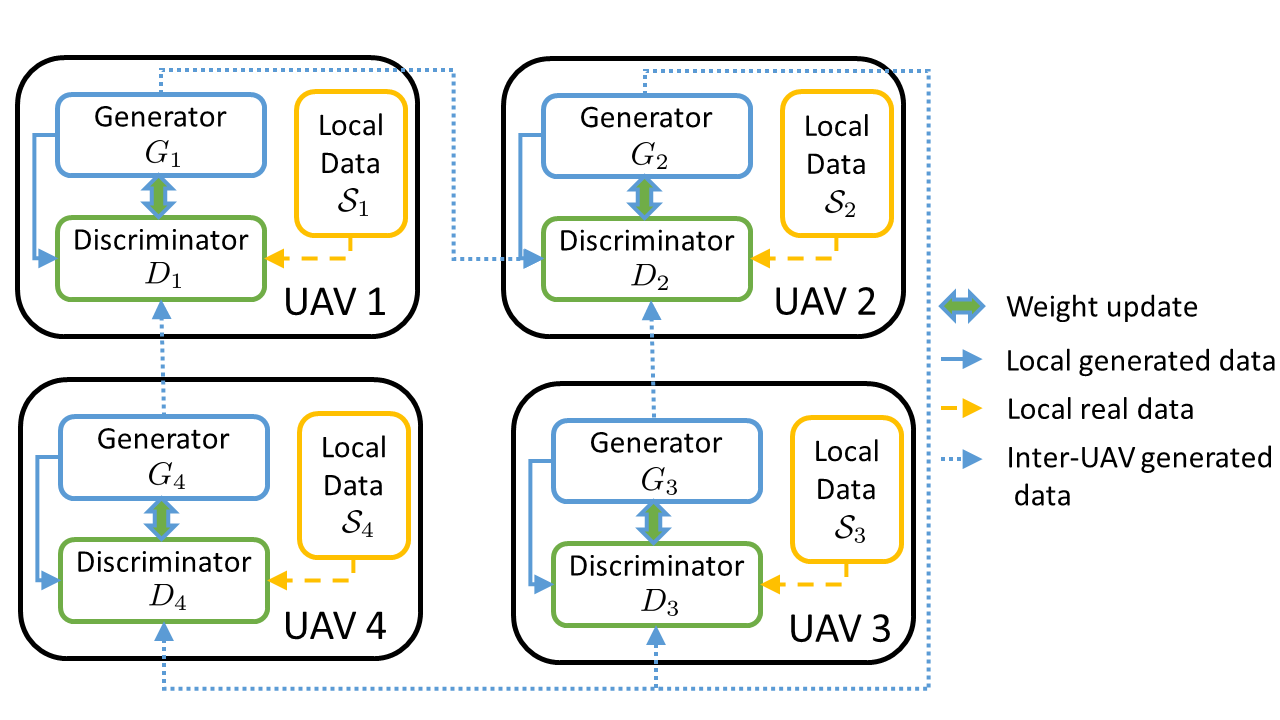}
		\vspace{-0.5cm}
		\caption{\label{learningFramework}\small An illustration of the distributed GAN framework with four UAVs, where $\mathcal{O}_1 = \{2\}$, $\mathcal{O}_2 = \{3,4\}$, $\mathcal{O}_3 = \{2\}$, and $\mathcal{O}_4 =\{1\}$. }  
	\end{center}\vspace{-0.8cm}  
\end{figure}

The goal is to train the generator distribution $f^g_{i}$ of each UAV $i$  to find the entire channel distribution $f$, under the constraint that no UAV $i$ sends its real dataset $\mathcal{S}_i$ or its DNN parameters $\boldsymbol{\theta}^g_i$ and $\boldsymbol{\theta}^d_i$ to  other UAVs. 
Instead, as shown in  \cite{ferdowsi2020brainstorming}, each UAV $i$ only shares the generated samples from $G_i$ in each training iteration.
Fig. \ref{learningFramework} illustrated the proposed distributed GAN framework, where the input of the discriminator for each UAV $i$ consists of the real samples from the local dataset $\mathcal{S}_i$, and the generated samples from the local generator $G_i$ and the generators of  its neighboring UAVs. 
In the distributed GAN framework  \cite{ferdowsi2020brainstorming}, the generators collaboratively generate channel samples to fool all of the discriminators while the discriminators try to distinguish between the generated and real channel samples.  
Let $\mathcal{N}_i$ be the set of  UAVs from whom UAV $i$ receives generated samples, and
let $\mathcal{O}_i$ be the UAV set to whom UAV $i$ sends its generated  samples. 
Then, for each UAV $i$ , the interaction between its generator and discriminator can be modeled by a game-theoretic framework with a value function: 
\begin{equation} \label{indU} \vspace{-0.1cm}
\begin{aligned}
V_i(D_i, G_i, \{ G_j \}_{j \in \mathcal{N}_i}) = &\mathbb{E}_{\boldsymbol{s}\sim f_i^b }[\log D_i(\boldsymbol{s})] +\\
& \mathbb{E}_{\boldsymbol{z} \sim f^z_i}[\log(1-D_i(G_i(\boldsymbol{z})))],
\end{aligned}
\end{equation}
where $f^b_i$ is a mixture distribution of UAV $i$'s  local dataset $\mathcal{S}_i$   and received data from all neighboring UAVs in $\mathcal{N}_i$, and $f_i^z$ is the sampling distribution of the random input $\boldsymbol{z}$.   
Here, we define $f^b_i = \pi_i f_i +  \sum_{j \in \mathcal{N}_i} \pi^g_{ij}f^g_j$, 
where $\pi_i =  S_i/(S_i + \eta \sum_{j \in \mathcal{N}_i}S_j)$,  $\pi^g_{ij}=  {\eta S_j}/{(S_i + \eta \sum_{j \in \mathcal{N}_i}S_j )}$, and $\eta {S}_j$ is the number of generated samples that  UAV $j$ transmits to UAV $i$ in each iteration, with $\eta >0$. 
Thus, the first term in (\ref{indU}) forces the discriminator  to output one for local real data and channel samples from other UAVs, and the second term penalizes generated data samples created by the local generator. 
Therefore, the generator of each UAV  aims to minimizing the value function, while the discriminator tries to maximize this value. Thus, the local training within each UAV between its generator and  discriminator forms a zero-sum game, and the total utility function of the  distributed GAN  network is  \cite{ferdowsi2020brainstorming}
\begin{equation}\label{Utility}\vspace{-0.1cm}
	V( \{D_i\}_{i=1}^{I},\{G_i\}^I_{i=1}) = \sum_{i=1}^I V_i(D_i, G_i, \{ G_j \}_{j \in \mathcal{N}_i}), 
\end{equation}
where all generators aim at minimizing the total utility function defined in (\ref{Utility}), while all discriminators try to maximize this value. 
Therefore, based on \cite{ferdowsi2020brainstorming}, the optimal discriminators and generators can be derived as a min-max problem as follows:
\begin{equation}\label{optimalDG}\vspace{-0.1cm}
	\{D^{*}_i\}_{i=1}^{I},\{G^{*}_i\}^I_{i=1} = \arg \min_{G_1,\cdots,G_I} \arg \max_{D_1,\cdots,D_I} V.  
\end{equation} 
Note that,  the optimal discriminators and generators in (\ref{optimalDG}) for the distributed GAN learning 
depend on the structure of the UAV communication system, which is defined as next.

\subsection{UAV Communication Network}

The communication structure of the UAV network is denoted by a directed graph $\mathcal{G} = (\mathcal{I},\mathcal{E})$, where $\mathcal{I}$ is the set of  UAVs, and $\mathcal{E}$ is the set of edges. 
Each edge $e_{ij}$ is an ordered UAV pair that corresponds to an air-to-air (A2A) communication link. 
For example, for any $i,j\in\mathcal{I}$, if $ e_{ij}\in \mathcal{E}$, then in each learning iteration, UAV $i$ will send its generated  data to the discriminator of UAV $j$.  
Meanwhile, for any $u,v \in\mathcal{I}$, if we can start from $u$, follow a set of connected  non-repeated edges in  $\mathcal{E}$, and finally reach $v$, then we say that a  path ${E}_{u,v}$ exists from $u$ to $v$, and the length $l_{u,v}$ equals to the number of edges on ${E}_{u,v}$.   

In order to efficiently share the generated channel samples, 
orthogonal frequency-division multiple access (OFDMA) techniques with $I$  available resource blocks (RBs) are used to support the A2A communication  over sub 6-GHz frequencies \cite{semiari2019integrated}. 
In order to avoid communication interference, we require  the number of communication links not to exceed the number of RBs, i.e., $|\mathcal{E}|\le I$, which is reasonable for UAV networks. 
Meanwhile, assuming a fixed hovering location for each UAV during the learning stage, the A2A communication rate from UAV $i$ to  $j$ using RB $b$ is given as 
$ R_{ij} = w_b \log_2(1+ P_{ij}h_{ij}/\sigma^2 )$, 
where $w_b$ is the A2A communication bandwidth, $P_{ij}$ and $h_{ij}$ are  the transmit power and path loss from UAV $i$ to  $j$,  and $\sigma^2$ is the noise power. 
A signal-to-noise ratio (SNR) threshold $\tau$ is introduced, such that for any UAV pair $(i,j)$, if the received  SNR at UAV $j$ is lower than $\tau$, i.e. $P_{ij}h_{ij}/\sigma^2 < \tau$, then, no RB will be assigned to this A2A communication link, i.e., $e_{ij} \notin \mathcal{E}$.  
In each iteration,  
each UAV $i$ sends $\eta {S}_i$ generated samples to its neighbors  $\mathcal{O}_i$, 
and the transmission time for  the generated channel samples  should not exceed  $t_\tau$. 


Here, we define the \emph{convergence time} $C$ of the distributed GAN approach as  the expected number of iterations that is required for the learning process to converge, multiplied by the duration of each learning iteration.    
To facilitate the analysis, we consider  a fixed size for each UAV’s dataset, i.e. $S_1 = \cdots = S_I$, and  a homogeneous UAV communication network, where $N_1= \cdots = N_I = N $. 
Then, the probability that the learning process converges after iteration $T$ can be derived as follows.    
\begin{theorem}\label{prob}  
	 Given the UAV network structure $\mathcal{G}$, the probability $p_{\mathcal{G}}(T)$ that the generator distribution $f^g_i$ of each UAV $i$  covers the entire data distribution $f$ after the $T$-th iteration can be given, based on the maximum shortest-path  $l^{\textrm{max}}$ in $\mathcal{G}$, as
	 \begin{equation}
	 \begin{aligned}
	 &p_{\mathcal{G}}(T) = \mathbbm{1}_{T \ge l^{\textrm{max}}} \quad \frac{\eta^{l^{\textrm{max}}}}{(1+N\eta)^{l^{\textrm{max}}-1}} \quad  +   \\
	 &\mathbbm{1}_{T > l^{\textrm{max}}}  \sum_{i=l^{\textrm{max}}+1}^T \left[  \prod_{j=l^{\textrm{max}}}^{i-1}\left(1- \frac{\eta^{l^{\textrm{max}}}}{(1+N\eta)^{j-1}}  \right)\right]  \frac{ \eta^{l^{\textrm{max}}}}{(1+N\eta)^{i-1}}. 
	 \end{aligned}	 	 
	 \end{equation}
\end{theorem}
\begin{proof}
	Proof is available in \cite{zhang2021distributed}.
\end{proof}
Theorem \ref{prob} shows that the convergence  iteration  is greater than or equal to the maximum shortest-path length $l^{\textrm{max}}$. 
This implies that to optimize the convergence rate for data sharing and channel modeling  in the UAV  network, it is necessary to minimize the maximum length of shortest paths among all UAVs. 
Then, based on Theorem \ref{prob}, the convergence  iteration $T_\mathcal{G} \in \mathbb{N}^+$  with a confidence level $p_\tau \in(0,1)$ is given by 
\begin{equation} 
p_{\mathcal{G}}(T_\mathcal{G}-1)< p_\tau \le    p_{\mathcal{G}}(T_\mathcal{G}). 
\end{equation}
That is, after  $T_\mathcal{G}$  iterations,  the generator distribution of each UAV is guaranteed to cover the entire channel distribution with a probability  $p_\tau$.    
Meanwhile, we assume  the local adversarial training between the generator and discriminator within each UAV to be perfect with a constant time cost $t_{c}$. 
Then, given the network structure $\mathcal{G}$, the overall convergence time of the distributed GAN learning is 
$C(\mathcal{G}) = (t_\tau + t_c) \cdot T_\mathcal{G} $. 

Consequently, in the distributed UAV network with limited communication resources, the objective for the cooperative mmWave channel modeling is to form an optimal A2A communication network $\mathcal{G}$, such that the expected convergence time of the distributed GAN learning is minimized, i.e., 
\begin{subequations}\label{equsOpt}  
	\begin{align}
	\min_{\mathcal{G}} \quad &  C(\mathcal{G})  \label{equOpt}\\ 
	\textrm{s. t.} \quad  
	& \sum_{e_{ij} \in \mathcal{E}} P_{ij} \le P_\textrm{max},  \quad \forall i \in \mathcal{I},  \label{consPower}\\ 
	& P_{ij}h_{ij}/\sigma^2 \ge \tau, \quad \forall e_{ij} \in \mathcal{E}, \label{consSNR}\\
	& \eta S_i/R_{ij}  \le t_\tau,  \quad \forall e_{ij} \in \mathcal{E},  \label{consTime} \\
	& \exists E_{i,j}  \subset \mathcal{E}  , \quad \forall i,j \in \mathcal{I}, \label{consPath} \\
	& |\mathcal{E}| \le I.  \label{consEdge} 	
	\end{align}
\end{subequations}
Here, (\ref{consPower}) limits the maximum transmit power  $P_\textrm{max}$  for each UAV,   (\ref{consSNR}) and (\ref{consTime}) set thresholds for the received SNR and the transmission time of each A2A communication link,  (\ref{consPath}) requires a strongly connected network in $\mathcal{G}$ such that each local channel dataset can be learned by all the other UAVs via the distributed GAN framework, and (\ref{consEdge}) avoids the interference over A2A communication links. 
Note that, in order to solve (\ref{equsOpt}),  a central controller is required to optimize the  communication structure based on the path loss information between each UAV pair. 
However, in the distributed UAV network,  such a centralized entity is often not available, which makes (\ref{equsOpt}) very challenging to solve.

\section{Optimal learning for distributed GANs } \label{solution}

\subsection{Optimal network structure for A2A UAV communications }
In order to optimally solve (\ref{equsOpt}) in a distributed manner without a central controller, we derive the graphic  property for the UAV network structure, 
based on (\ref{consPath}) and (\ref{consEdge}), 
 as follows. 
\begin{theorem}\label{gStructure}
	Under the constraint that the number of communication edges is smaller than or equal to the number of UAVs, the strongly connected network must have a ring structure, i.e.,   $N_i = O_i = 1$, $\mathcal{N}_i \cap \mathcal{N}_j = \emptyset$, and $\mathcal{O}_i \cap \mathcal{O}_j = \emptyset$,   $\forall i,j \in \mathcal{I}$.
\end{theorem}
\begin{proof}
	Proof is available in \cite{zhang2021distributed}.
\end{proof}
Theorem \ref{gStructure} shows that, given constraints (\ref{consPath}) and (\ref{consEdge}), the network structure of the UAV communication system must be a ring, where each UAV receives the channel sample from one UAV, and sends its generated data to another UAV.  

Based on Theorems \ref{prob} and \ref{gStructure}, we can equivalently reformulate (\ref{equsOpt}) into a set of distributed optimization problems, such that the objective of each UAV $i$ is to choose the optimal single UAV $ \mathcal{O}_i =  \{ o_i \}$ to whom UAV $i$ sends its generated channel samples, so as to minimize the convergence time over its maximum shortest-path while satisfying constraints (\ref{consPower})-(\ref{consTime}), i.e.,
\begin{subequations}\label{equsOpt2}  
	\begin{align}
	\min_{o_i \in \mathcal{I}_{-i}} \quad &  l_i^{\textrm{max}} (\mathcal{G} + e_{i,o_i})  \label{equOpt2}\\ 
	\textrm{s. t.} \quad  
	& P_{i,o_i} \le P_\textrm{max},    \label{consPower2}\\ 
	& P_{i,o_i}h_{i,o_i}/\sigma^2 \ge \tau,   \label{consSNR2}\\
	& \eta S_i/R_{i,o_i}  \le t_\tau,  \label{consTime2}  	
	\end{align}
\end{subequations}
where $\mathcal{I}_{-i}$ is the set of UAVs except for  $i$, $\mathcal{G} + e_{i,o_i}$ is the  graph structure generated by adding an edge $e_{i,o_i}$ to $\mathcal{G}$, and $l_i^{\textrm{max}} $ is the  maximum shortest-path  from UAV $i$ to any other UAVs. 
We define the set of feasible UAVs  to whom  UAV $i$ can send its generated channel samples  while satisfying constraints (\ref{consPower2})-(\ref{consTime2})  as 
$\mathcal{J}_i = \{ j\in \mathcal{I}_{-i} | P_{ij} \le P_\textrm{max},   P_{ij} h_{ij}/\sigma^2 \ge \tau, \eta S_i/ R_{ij} \le t_\tau  \}$. 
Then, the necessary condition for a feasible solution to (\ref{equsOpt2}) is  provided next.   
\begin{proposition}[Necessary condition]\label{existence}
	 A feasible topology solution to  (\ref{equsOpt2}) exists,  only if  $\bigcup_{i=1}^{I} \mathcal{J}_i = \mathcal{I}$ and $\forall i,\mathcal{J}_i \ne \emptyset$ hold.  
\end{proposition} 
\begin{proof}
	Proof is available in \cite{zhang2021distributed}.
\end{proof} 
Proposition \ref{existence} shows that if the union of feasible sets does not cover all UAVs, then the UAV network cannot form a strongly connected graph and a feasible solution to  (\ref{equsOpt2})  does not exist.      
Based on Proposition \ref{existence} and Theorem \ref{gStructure}, 
we derive the sufficient condition for the optimal network structure  that maximizes the convergence rate for the distributed GAN learning  as follows. 
\begin{proposition}[Sufficient condition]\label{optimalResult}
	Given that $~~\bigcup_{i=1}^{I} \mathcal{J}_i = \mathcal{I}$ and $\mathcal{J}_i \ne \emptyset$ hold for all $i\in\mathcal{I}$,  the optimal UAV network structure is  $\mathcal{G}^{*} = (\mathcal{I},\mathcal{E})$, where $\mathcal{E} \subseteq \{ e_{ij}| i\in \mathcal{I}, j\in \mathcal{J}_i  \}$ and $l_i^{\textrm{max}}(\mathcal{G}^{*}) = I-1$, $\forall i \in \mathcal{I}$.  
\end{proposition} 
\begin{proof}
Proof is available in \cite{zhang2021distributed}.
\end{proof} 

Consequently,  the optimal UAV network   $\mathcal{G}^{*}$ that minimizes the convergence time $C(\mathcal{G}^{*})$ has a ring  structure with a  communication link set  $\mathcal{E} \subseteq \{ e_{ij}| i\in \mathcal{I}, j\in \mathcal{J}_i  \}$.

\subsection{Optimal learning for distributed GANs}

Based on the optimal network structure $\mathcal{G}^{*}$, according to \cite[Proposition 1 and  Theorem 1]{ferdowsi2020brainstorming}, the optimal generator of each UAV $i$ for the distributed GAN learning  is $G_i^{*} \sim {f^{g}_{i}}^* = f^{b}_{i} = (f_i +  f^g_j)/2$ for $e_{ji} \in \mathcal{E}$, and the optimal discriminator is $D_i^{*}$ = ${f^{b}_{i}}/({f^{b}_{i} + {f^{g}_{i}}^*})$ $=$ $0.5$.  
That is, for each UAV $i$, its generator's distribution ${f^{g}_{i}}^*$  equals to the mixture of the channel distribution $f_i$ from its local dataset $\mathcal{S}_i$ and the generator's distribution  $f^g_j$ from neighboring UAV $j$, and, thus, the discriminator cannot distinguish the generated channel samples from the real data. In this case, the learning process in the UAV network converges to a Nash equilibrium (NE), and the generator of each UAV $i$ learns the entire distribution of mmWave channels, i.e., $G_i^{*} \sim f$.   
The  formation approach of the optimal UAV network, as well as the distributed GAN learning algorithm for mmWave channel modeling, is summarized in Algorithm 1.

\begin{algorithm}[t] \small   
	\caption{UAV network formation with distributed GAN learning for mmWave channel modeling} \label{algo}
	\begin{algorithmic}
		\State \textbf{UAV Network Formation:} \\
		1. Each UAV $i$ uses its own RB to measure channel $h_{ij}$ for $j \in \mathcal{I}_{-i}$, \\ 
		\quad  and broadcasts the feasible UAV set $\mathcal{J}_i$;\\
		2. If $\bigcup_{i=1}^{I} \mathcal{J}_i = \mathcal{I}$ and $\bigcap_{i=1}^{I} \mathcal{J}_i \ne \emptyset$, go to step 3; otherwise, the\\
		\quad UAVs need to adjust their locations, and then, go back to step 1; \\  
		3. Start with the network graph where $\mathcal{E} =  \{ e_{ij}| i \in \mathcal{I}, j \in \mathcal{J}_i \} $; \\
		4. \textbf{For} each UAV $i$ with $|\mathcal{J}_i|>1$, \\
		\quad \quad Remove one edge $e_{ij}$ from  $\mathcal{E}$ where $j \in \mathcal{J}_i$, while guaranteeing \\
		\quad \quad $(\bigcup_{k\in\mathcal{I}_{-i}} \mathcal{J}_k) \cup (\mathcal{J}_i-j )= \mathcal{I}$;\\
		\quad \textbf{Until} $|\mathcal{J}_i|=1$ for all $i \in \mathcal{I}$. \\ 
	\\
		\textbf{Distributed GAN learning:} \\
		A. Initialize $G_i$ and $D_i$ for each UAV $i \in \mathcal{I}$; \\
		B. \textbf{Repeat:} Parallel for all  $i \in \mathcal{I}$: \\ 
		\quad a. Sample $\pi_iu$ real channel samples: $\boldsymbol{s}_i^{(1)}, \cdots, \boldsymbol{s}_i^{(\pi_iu)} \sim \mathcal{S}_i$;\\
		\quad b. Generate $u$ channel samples $G_i(\boldsymbol{z}^{(1)}), \cdots, G_i(\boldsymbol{z}^{(u)}) $ from $f^G_i$ \\
		\quad \quad  and $f^z_i$; \\
		\quad c. Send $\pi_{oi} u$ generated data to each UAV $o \in \mathcal{O}_i$,  and receive \\ 
		\quad \quad $\pi_{ij}u$ data samples $\boldsymbol{s}_j^{(1)}, \cdots, \boldsymbol{s}_j^{(\pi_{ij}u)}$ from each UAV $j \in \mathcal{N}_i$; \\
		\quad d. Update $\boldsymbol{\theta}_i^d$ via gradient ascent:  \\
		\quad \quad $\nabla_{\boldsymbol{\theta}_i^d} V(D_i(\boldsymbol{\theta}_i^d)) =\frac{1}{2u} \nabla_{\boldsymbol{\theta}_i^d}[\sum_{k=1}^{\pi_iu} \log(D_i(\boldsymbol{s}_i^{(k)})) +$ \\
		\quad \quad $\sum_{j \in \mathcal{N}_i} \sum_{k=1}^{\pi_{ij}u}\log(D_i(\boldsymbol{s}_j^{(k)}))  + \sum_{k=1}^{u}  \log(1- D_i(G_i(\boldsymbol{z}^{(k)})))] $;\\
		\quad e. Update $\boldsymbol{\theta}_i^g$ via gradient descent: \\
		\quad \quad \quad $\nabla_{\boldsymbol{\theta}_i^g} V(G_i(\boldsymbol{\theta}_i^g)) = \frac{1}{u} \nabla_{\boldsymbol{\theta}_i^g} \sum_{k=1}^{u} \log(1- D_i(G_i(\boldsymbol{z}^{(k)}))) $; \\
		\quad \textbf{Until}  convergence to the NE. 
	\end{algorithmic}
\end{algorithm}

\section{Simulation Results and Analysis}\label{simulation}

For our simulations, we consider an airborne  network with four UAVs that provide wireless service within a geographic area of  $100 \times 100~ \textrm{m}^2$. Each UAV has a mmWave channel dataset \cite{alkhateeb2019deepmimo} that covers one of four regions in the area without overlap,  i.e., business blocks, residential areas, rural region and a city park.  
For simulation parameters, we set $M=256$, $N=64$, $f = 30$ GHz, $w_b = 2$ MHz, $P_{\textrm{max}} = 40$ dBm, $\sigma^2 = -174$ dBm/Hz, $\tau = 10$ dB,  $t_\tau = 0.1$ second, $\eta = 1.4$, and $S_i = 1000$ for each UAV $i$. 
We implement a neural network (NN) with two convolutional layers for the GAN discriminator, and another NN with two transposed convolutional layers   for the generator. 

Fig. \ref{CT} shows the convergence rates of the distributed GAN learning, for different sizes of shared data samples $\eta$ and for different numbers of UAVs, respectively.  
Note that, in each iteration, each UAV $i$ sends $\eta S_i$  generated channel samples to its neighboring UAVs in $\mathcal{O}_i$. 
As shown in the upper plot of Fig. \ref{CT}, when $\eta$ becomes larger,  the convergence rate of our distributed GAN approach becomes faster. 
Given that the maximum path length $l^{\textrm{max}}$ for a four-UAVs network equals to three,  the convergence probabilities  remains to be zero, until  $T$ is  equal to or greater than three. 
Meanwhile,  our proposed algorithm shows a rapid convergence property when $\eta=1.4$, and the distributed GAN learning converges with a probability of over $90\%$ after six iterations.   
Next, we show the relationship between the convergence rate and the number of UAVs in the lower plot of Fig. \ref{CT}, for a fixed generated sample size $\eta = 1.4$.  
For larger network sizes, the convergence rate of the channel modeling process decreases, due to a longer path length in the distributed learning system. 
Therefore, in a large UAV network, the size $\eta$ of the generated channel samples needs to be adaptively adjusted to guarantee an efficient learning. 

\begin{figure}[!t]
	\begin{center}
		\vspace{-0.12cm}
		\includegraphics[width=7.8cm]{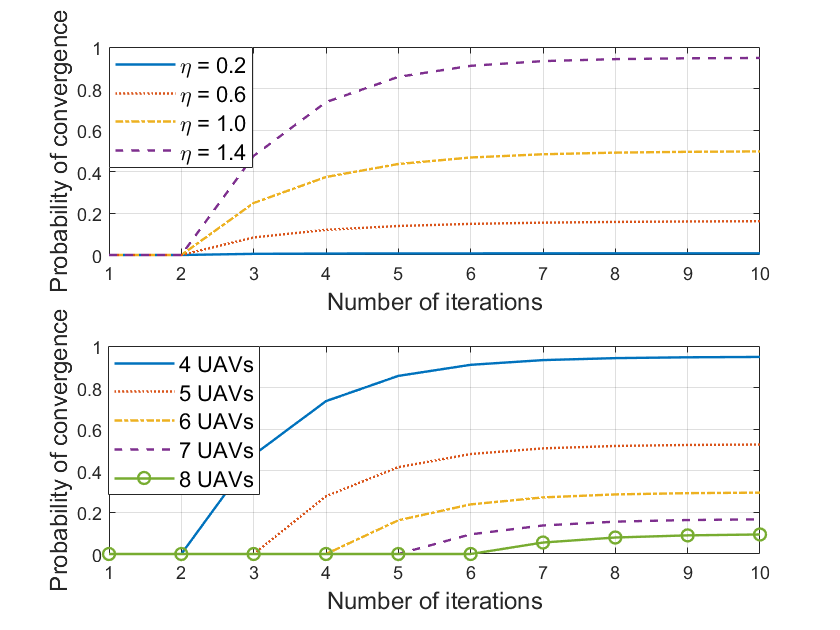}
		\caption{\label{CT}\small Convergence rate for different sizes of shared samples $\eta$ and different network sizes. }  
	\end{center}\vspace{-0.4cm}  
\end{figure}
\begin{figure}[!t]
	\begin{center}
		\vspace{-0.4cm}
		\includegraphics[width=7.8cm]{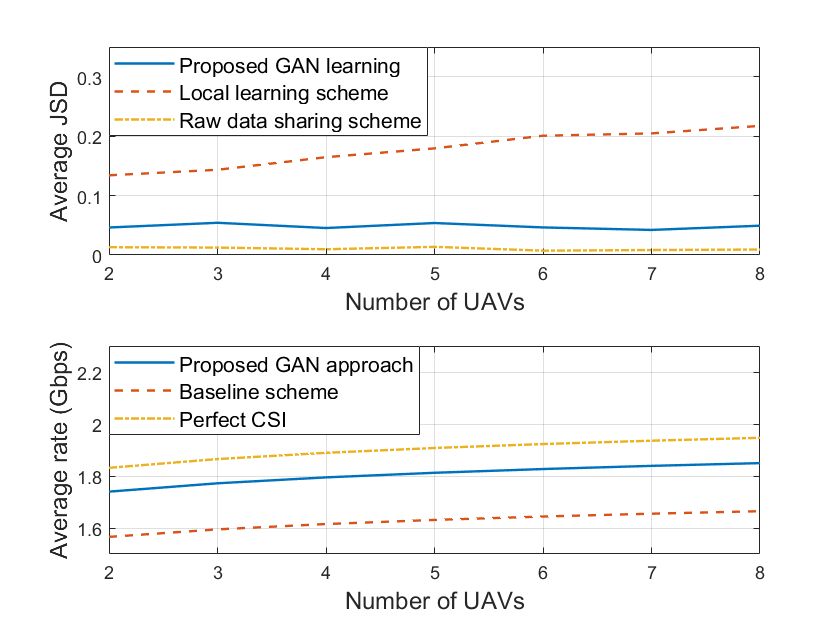}
		\caption{\label{jsd}\small Average JSD($f$,$f^g_i$) of channel modeling and  average data rate of UAV downlink communications for different network sizes. }    
	\end{center}\vspace{-0.6cm}  
\end{figure}

In Fig. \ref{jsd}, we evaluate the channel modeling accuracy and the communication performance of the proposed distributed GAN algorithm. 
First, in the upper plot of Fig. \ref{jsd}, we introduce a baseline scheme that performs local channel modeling without information sharing  and a distributed learning scheme that shares raw channel data between each UAV, and Jensen-Shannon divergence (JSD) is used as the performance metric, where a lower value of JSD indicates a higher learning accuracy.  
Fig. \ref{jsd} first shows that the modeling accuracy of the proposed distributed GAN approach  outperforms the local learning scheme. 
Given more UAVs in the system, each UAV covers a smaller service area, and the local generator distribution only applies to a limited spatial domain. Thus, the modeling accuracy of the local learning scheme decreases for  a larger network size.  
However, using the distributed learning approach, each UAV can learn the A2G channel property over a larger  location domain from the generated samples of other UAVs. Thus, the modeling accuracy for the proposed approach stays the same for different network sizes. 
Moreover, due to a limited training time and the inevitable training error at each UAV, the overall distributed GAN training of the UAV network may converge to a local optimum.  
This explains the  performance gap between the proposed learning scheme and  the raw data sharing scheme.    
In the lower plot of Fig. \ref{jsd}, we evaluate the time-average data rate of the UAV A2G communications with a $50$ MHz bandwidth, using the proposed channel modeling approach and two other schemes:  
A baseline scheme that requires a constant pilot training on mmWave channels, and an upper-bound scheme that assumes a known CSI. 
Fig. \ref{jsd} shows that, given more UAVs in the network, the average data rates of all three schemes increase, due to an averagely smaller service area for each UAV. 
Our proposed method applies the trained channel model for downlink transmissions, thus avoiding constant channel estimation. 
Therefore, compared with the real-time measurement scheme, the proposed method improves the time-average data rate by over $10\%$. 
However, due to the  inevitable training error in the proposed channel model,  our distributed GAN  method yields a lower data rate, compared with a perfect CSI scheme.  

\section{Conclusion}\label{conclusion}

In this paper, we have proposed a novel framework for mmWave channel modeling  in a UAV cellular network.  
Based on the distributed GANs, a cooperative learning framework has been developed for each UAV to learn the mmWave channel distribution from other agents in the privacy-preserving and  distributed manner.  
We have derived the necessary and sufficient conditions for the optimal network structure of information sharing that maximizes the learning rate. 
Simulation results have shown that  the learning rate will increase by sharing more generated samples, but decrease given a larger UAV network size. 
The results also show that the proposed distributed GAN approach yields a higher learning accuracy, compared with a standalone GAN, and it improves the average data rate of UAV downlink communications by over $10\%$, compared with the baseline real-time channel estimation scheme. 

\bibliographystyle{IEEEtran}
\bibliography{references}

\end{document}